\documentclass[a4paper,10pt]{article}
\usepackage[utf8]{inputenc}
\usepackage{authblk}
\usepackage{float}
\usepackage{amsthm}
\usepackage{amsfonts}
\usepackage{amsmath}
\usepackage{tcolorbox}
\usepackage{multirow}
\usepackage{todonotes}
\usepackage{appendix}
\usepackage[linesnumbered,ruled]{algorithm2e}
\usepackage[utf8]{inputenc}
\newtheorem{theorem}{Theorem}
\newtheorem{lemma}{Lemma}

\newtheorem{definition}{Definition}
\newtheorem{corollary}{Corollary}
\newtheorem{observation}{Observation}
\newtheorem{proposition}{Proposition}

\usepackage{hyperref}
\hypersetup{
    colorlinks=true,
    linkcolor=magenta,
    filecolor=blue,      
    urlcolor=blue,
    citecolor=blue,
    pdftitle={Sharelatex Example},
    }

\usepackage{geometry}

\newcommand{\fpt}{{\sf FPT}\xspace}
\newcommand{\npc}{\textsf{NP}-complete\xspace}
\newcommand{\nph}{\textsf{NP}-hard\xspace}

\newcommand{\woh}{\textsf{W[1]}-hard\xspace}

\newcommand{\yes}{\texttt{Yes}\xspace}
\newcommand{\no}{\textsf{No}\xspace}

\newcommand{\ccC}{{\mathcal C}}
\newcommand{\oh}{{\mathcal O}}

\newcommand{{\DAGstLengthWeight}}{{\sc $(s,t)$-DAG-Length-and-Edge-Weight}}

\newcommand{\spfgp}{{\sc Shortest Path with Forcing Graph}\xspace}

\newcommand{\OO}{{\mathcal O}}
\newcommand{\GG}{{\mathcal G}}

\newcommand{\defprob}[3]{
	\begin{tcolorbox}[colback=gray!5!white,colframe=gray!75!black]
		\vspace{-1mm}
			\begin{tabular*}{\textwidth}{@{\extracolsep{\fill}}lr} #1   
            \vspace{2mm} \\ \end{tabular*}
			{\bf{Input:}} #2  
            \vspace{2mm}\\
			{\bf{Goal:}} #3
		\end{tcolorbox}
	}

\newcommand{\defparprob}[4]{
\begin{tcolorbox}[colback=gray!5!white,colframe=gray!75!black]
  \vspace{-1mm}
  \begin{tabular*}{\textwidth}{@{\extracolsep{\fill}}lr} #1  & {\bf{Parameter:}} #3 \\ \end{tabular*}
  {\bf{Input:}} #2  \\
  {\bf{Question:}} #4
  \vspace{-1mm}
\end{tcolorbox}
}

\title{On the Polynomial Kernelizations of Finding a Shortest Path with Positive Disjunctive Constraints\thanks{A preliminary version of this paper has appeared in Proceedings of COCOA-2024 \cite{BandopadhyayBMP24}. Research of Diptapriyo Majumdar has been supported by Science and Engineering Board Research (SERB) grant SRG/2023/001592.}}

\author[1]{Susobhan Bandopadhyay}
\author[2]{Suman Banerjee}
\author[3]{Diptapriyo Majumdar}
\author[4]{Fahad Panolan}
\affil[1]{Tata Institute of Fundamental Research, Mumbai, India\\ susobhan.bandopadhyay@tifr.res.in}
\affil[2]{Department of Computer Science and Engineering, Indian Institute of Technology, Jammu, India\\suman.banerjee@iitjammu.ac.in}
\affil[3]{Indraprastha Institute of Information Technology Delhi, New Delhi, India\\ diptapriyo@iiitd.ac.in}
\affil[4]{School of Computer Science,  University of Leeds,  United Kingdom\\ 
F.Panolan@leeds.ac.uk}
\date{}

\begin{document}

\maketitle

\begin{abstract}
     We study the {\sc Shortest Path} problem subject to positive binary disjunctive constraints. 
In positive disjunctive constraints, there are certain pairs of edges such that at least one edge from every pair must be part of every feasible solution. 
We initiate the study of {\sc Shortest Path} with binary positive disjunctive constraints in the perspective of parameterized complexity.
Formally, the input instance is a simple undirected graph $G = (V, E)$, a forcing graph $G_f = (E, E')$, two vertices $s, t \in V(G)$ and an integer $k$.
Note that the vertex set of $G_f$ is the same as the edge set of $G$.
The goal is to find a set $S$ of at most $k$ edges from $G$ such that there is a path from $s$ to $t$ in the subgraph $G = (V, S)$ and $S$ is a vertex cover in $G_{f}$.
In this paper, we consider two different natural parameterizations for this problem.
One natural parameter is the solution size, i.e. $k$ which provides polynomial kernelization results.
The other natural parameter is structural parameterizations of $G_f$, i.e. the size of a modulator $X \subseteq E(G) = V(G_f)$ such that $G_f - X$ belongs to some hereditary graph class.
We discuss the parameterized complexity of this problem under some structural parameterizations.    
 \end{abstract}
 
\section{Introduction}
\label{sec:intro}

In recent times, several classical combinatorial optimization problems on graphs including maximum matching, shortest path, Steiner tree, etc. have been studied along with some additional binary conjunctive and/or disjunctive constraints  \cite{agrawal2020parameterized,darmann2011paths}.
Darmann et al. \cite{darmann2011paths} have studied finding shortest paths, minimum spanning trees, and maximum matching of graphs with binary disjunctive constraints from the perspective of classical complexity.
These positive or negative binary constraints are defined with respect to pairs of edges.
\begin{itemize}
	\item A {\em negative disjunctive constraint} between an edge-pair $e_i$ and $e_j$ says that both $e_i$ and $e_j$ cannot be present in a feasible solution.
	\item A {\em positive disjunctive constraint} between an edge pair $e_i$ and $e_j$ says that either the edge $e_i$ or the edge $e_j$ or both must be present in any feasible solution.
\end{itemize}

In the graph theoretic problems like {\sc Maximum Matching, Minimum Spanning Tree, Shortest Path} etc, the objective is to select a subset of edges from a {\em primary graph} satisfying certain properties.
The negative disjunctive constraint can be interpreted as a {\em conflict graph} such that each vertex of the conflict graph corresponds to an edge in the {primary graph}.
Furthermore, for every edge in the conflict graph, at most one endpoint can be part of any feasible solution.
Thus, a feasible solution must be an independent set in the conflict graph.
The positive disjunctive constraints can be interpreted as a {\em forcing graph} such that each vertex of the forcing graph corresponds to an edge in the {primary graph}.
In the forcing graph, each edge must have at least one endpoint included in any feasible solution.
Therefore, in the case of positive disjunctive constraints, a feasible solution must be a vertex cover in the forcing graph.
Formally, an input to the \textsf{Forcing-Version} of a classical combinatorial optimization problem $\Pi$, called as \textsf{Forcing-Version} $\Pi$ consists of an instance $I$ of $\Pi$ along with a forcing graph $G_{f}$; i.e.; $(I, G_{f})$.
The vertex set of the forcing graph is the edge set of the {primary graph}. 
A solution of \textsf{Forcing-Version} $\Pi$ for the instance $(I, G_{f})$ is a solution of $I$ for the original problem, along with the property that the solution forms a vertex cover for $G_{f}$.
To the best of our knowledge, none of the problems {\sc Shortest Path}, {\sc Maximum Matching}, {\sc Minimum Spanning Tree} have been explored with positive disjunctive constraints from the perspective of parameterized complexity.

\medskip

In this paper, we initiate the study of {\sc Shortest Path} problem with positive disjunctive constraints from the perspective of parameterized complexity and kernelization (see Section 2 for definitions, etc).
Formally, in the  \textsc{{\spfgp} (SPFG)} problem, we are given a simple, unweighted graph $G(V,E)$, two vertices $s$ and $t$, a positive integer $k$ and a forcing graph $G_{f}(E, E^{'})$.
The decision version of this problem asks to check whether there exists a set $E^* \subseteq E(G)$ of at most $k$ edges such that the subgraph induced by the vertex set $V(E^*)$ in $G$ contains an $s$-$t$ path and also $E^*$ forms a vertex cover in $G_{f}$.
As ``solution size'' is the most natural parameter, we formally state the definition of the parameterized version of our problem as follows.
%
%
%

    \defparprob{\textsc{Shortest Path with Forcing Graph}}{
A simple, undirected graph $G(V,E)$, two distinct vertices $s,t \in V(G)$, a positive integer $k$, and a forcing graph $G_{f}(E, E^{'})$.
}{k}{
 Is there a set $E^*$ of at most $k$ edges from $G$ such that the subgraph $G(V, E^{*})$ contains an $s$-$t$ path in $G$, and $E^{*}$ forms a vertex cover in $G_{f}$?}

In the above definition, the considered parameter is `solution size'.
Darmann et al. \cite{darmann2011paths} have proved that {\spfgp} is NP-hard even when the forcing graph $G_f$ is a graph of degree at most one.
So, it can be concluded that even when the forcing graph $G_f$ is very sparse, then also {\spfgp} is NP-hard.
In the first part of our paper, we consider {\spfgp} when parameterized by solution size.

In addition, it is also natural to consider some parameters that are some structure of the input.
Observe that the solution size must be as large as the minimum vertex cover size of the forcing graph.
We also consider the `deletion distance of $G_f$ to some hereditary graph class $\GG$', such that this deletion distance is a parameter that is at most as large as the solution size whenever $\GG$ contains a graph that has at least one edge.
In the second part of our paper, we also initiate the study of this problem when the considered parameter is the deletion distance of $G_f$ to some special graph class.
Formally, the definition of this parameterized version is the following.

%
%

\defparprob{\textsc{SPFG-$\GG$-Deletion}}{
 A simple, undirected graph $G(V,E)$, two distinct vertices $s,t \in V(G)$, a positive integer $k$, a forcing graph $G_{f}(E, E^{'})$, a set $X \subseteq E$ such that $G_f - X \in \GG$.
}{$|X|$}{
 Is there a subset $E^*$ of at most $k$ edges from $G$ such that the subgraph $G(V, E^{*})$ contains an $s$-$t$ path in $G$, and $E^{*}$ forms a vertex cover in $G_{f}$?}
 
\paragraph{Our Contributions:} In this paper, we study the {\spfgp} under the realm of parameterized complexity and kernelization (i.e., polynomial-time preprocessing).
We first consider both the {primary graph} $G$ and the forcing graph $G_f$ to be arbitrary graphs and provide {\fpt} and kernelization algorithms.
Next, we initiate a systematic study on what happens to the kernelization complexity when either $G$ is a special graph class or $G_f$ is a special graph class.
Formally, we provide the following results for {\spfgp} when the solution size ($k)$ is considered as the parameter.

\begin{itemize}
	\item First, we prove that {\spfgp} admits an algorithm that runs in $\oh^{*}(2^k)$-time\footnote{$\oh^*$ hides polynomial factor in the input size.}. Moreover, we prove that {\spfgp} is polynomial-time solvable when $G_f$ is the class of all $2K_2$-free graphs (i.e., $C_4$-free graphs).
	\item Then, we prove our main result. In particular, we prove that {\spfgp} admits a kernel with $\oh(k^5)$-vertices when there are no restrictions on $G$ and $G_f$.
	\item After that, we consider the kernelization complexity of {\spfgp} when $G$ is a planar graph and $G_f$ is a general graph.
	In this condition, we provide a kernel with $\oh(k^3)$ vertices for {\spfgp}.
	
	\item Next, we consider when $G$ is an arbitrary graph and $G_f$ is a graph belonging to a special hereditary graph class. In this paper, we focus on the condition when $G_f$ is a cluster graph (i.e., a disjoint union of cliques) or a bounded degree graph.
	In both these conditions, we provide a kernel with $\oh(k^3)$ vertices for {\spfgp}.
\end{itemize} 

Finally, we consider the {\sc SPFG-${\GG}$-Deletion} problem.
It follows from the results of Darmann et al. \cite{darmann2011paths} that even if $X = \emptyset$ and $G_f$ is a $2$-ladder (also known as a collection of disjoint edges) {\sc SPFG-$\GG$-Deletion} is {\nph}. 
Therefore, it is unlikely to expect the possibility that {\sc SPFG-$\GG$-Deletion} would admit an \fpt algorithm even when $\GG$ is a very sparse graph class.
We complement their {\sf NP}-hardness result by proving that
\begin{itemize}
	\item {\sc SPFG-${\GG}$-Deletion} admits an {\fpt} algorithm when $\GG$ is the class of all $2K_2$-free graphs.
\end{itemize}

\paragraph{Related Work} Recently, conflict-free and forcing variants of several classical combinatorial optimization problems including Maximum Flow \cite{pferschy2013maximum}, Bin Packing \cite{epstein2011online,gendreau2004heuristics,jansen1999approximation} Knapsack \cite{pferschy2009knapsack,pferschy2017approximation}, Maximum Matching \cite{darmann2009determining}, Minimum Spanning Tree \cite{darmann2009determining,darmann2011paths}, Set Cover \cite{even2009scheduling}, Shortest Path \cite{darmann2011paths} has been studied extensively in both algorithmic and complexity-theoretic point of view. 
It has been observed that most of them are {\nph} even when the classical version is polynomial-time solvable.
 Recently, some of these problems have also been studied in the realm of parameterized complexity as well \cite{agrawal2020parameterized,jain2020conflict}. 
Agrawal et al. \cite{agrawal2020parameterized} {considered} the conflict-free version of the {\sc Maximum Matching} and the {\sc Shortest Path}.
 They showed that {both problems} are {\woh} when parameterized by solution size.
 They also investigated the complexity of the problems when the conflict graph has some topological structure.
 They gave an \textsf{FPT} algorithm for the conflict-free matching problem when the conflict graph is chordal.
 They also gave \textsf{FPT} Algorithms for {both the problems} when the conflict graph is $d$-generate.
 They also showed that {both problems} are in \textsf{FPT} when the conflicting conditions can be represented by matroids.
 Darmann et al. \cite{darmann2011paths} studied {both the problems} along with both the constraints conflict graph and forcing graph.
 They showed that the conflict-free variant of the maximum matching problem is {\nph} even when {the conflict graph is a $2$-ladder}.
 Similarly, the conflict-free variant of the shortest path problem is \textsf{APX}-hard even when the conflict graph is $2$-ladder.
 They also showed the minimum cost spanning tree problem is strongly {\nph} even when the conflict graph is a disjoint union of paths of length two or three.  
  
 \section{Preliminaries}
 In this section, we describe the notations and symbols used in this paper.
 Additionally, we also provide some basic results in this section.

\subsection{Graph Theory}
 \label{sec:graph-theory}
 
All the graphs considered in this paper are simple, finite, undirected, and unweighted.
{A {\em walk} in an undirected graph is a sequence $v_0, e_1, v_1, e_2, v_2,\ldots, e_k, v_k$ of vertices and edges such that edge $e_i$ has endpoints $v_{i-1}$ and $v_i$. 
An {\em $s$-$t$ walk} is a walk that has the first vertex $s$ and last vertex $t$.
An {\em $s$-$t$ path} is a walk that has no repeated vertex.
It is clear from the definition that a path cannot have any repeated edges.
We often denote a path $P$ by a sequence of vertices $v_1, v_2, \ldots, v_k$ such that every edge is $v_i v_{i+1}$ for every $i \in [k-1]$.}
Other notations and terminologies used in this paper are standard and adopted from Diestel's book of graph theory \cite{diestel2005graph}.
 In our problem, we are dealing with two different graphs: the {primary graph} $G$ and the forcing graph $G_{f}$.
 We denote the number of vertices and edges of $G$ by $n$ and $m$, respectively.
 Similarly, for $G_{f}$, it is $m$ and $m'$ (since the edge set of $G$ is the same as the vertex set of $G_f$).
 Given a graph $G(V, E)$ and an edge $e \in E(G)$, $V_{e}$ denotes the set of { the two endpoints} of $e$.
 For any subset of edges, $E^* \subseteq E(G)$,
 by $V(E^*)$ we denote the { set of all endpoints of the edges in $E^*$}; i.e.; $V(E^*)= \underset{e \in {E^*}}{\bigcup} \ V_{e}$.
 For every $u, v \in V(G)$, we denote its shortest path distance by $dist(u,v)$.
 For any graph $G(V, E)$, a subset of its vertices ${S} \subseteq V(G)$ is said to be a \emph{vertex cover} of $G$ if every edge of $G$ has at least one of its endpoints in ${S}$.
 A subset of the vertices ${S}$ is said to be an independent set of $G$ if, between any pair of vertices of $S$, there does not exist any edge in $G$.
 For any subset of vertices $S$ of $G$, the subgraph induced by the vertex set in $G$ is denoted by $G[S]$.
 Given an edge set $S \subseteq E(G)$, the graph $G[S]$ has vertex set $V(G)$ and the edge set $S$.
 For any graph $G$ and any vertex $v \in V(G)$, $G-\{v\}$ denotes the graph that can be obtained by deleting $v$ and the edges incident on it from $G$.
 This notion can be extended to a subset of vertices as well.
 Given a graph $G = (V, E)$ and a set $X \subseteq V(G)$, we define the {\em identification} of the vertex subset $X$ into (a new vertex) $u_X$ by constructing a graph $\hat G$ as follows.
 First, delete the vertices of $X$ from $G$ and add a new vertex $u_X$.
 Then, for every $v \in N_G(X)$, make $vu_X$ an edge of $\hat G$.
{ This graph operation has been defined in the literature (see \cite{ComasSerna2009,MajumdarRS20,MorelleST2026}).}
\medskip 

 A graph is said to be a \emph{cluster graph} if every connected component is a clique.
 A graph is said to be a {\em degree-$\eta$-graph} if every vertex has degree at most $\eta$.
 A connected graph is said to be {\em $2K_{2}$-free} if it does not contain any pair of edges that are nonadjacent to each other.
 A graph is said to be a {\em planar graph} if it can be drawn on the surface of a sphere without crossing edges.
 We use the following property of a planar graph in our results.
 
\begin{proposition}[\cite{WestGraphTheorybook}]
\label{prop:planar-graph-edges}
If $G$ is a simple planar graph with $n$ vertices, then $G$ has at most $3n - 6$ edges.
\end{proposition}

A graph is said to be a {\em $2$-ladder} if every connected component is a path of length one.
Similarly, a graph is said to be a {\em $3$-ladder} if every connected component is a path of length two.
 
  
\subsection{Parameterized Complexity and Kernelization}
\label{sec:FPT-kernels}
 
A {\em parameterized problem} $\Pi$ is denoted as a subset { of} $\Sigma^* \times {\mathbb{N}}$.
An instance of a parameterized problem is denoted by $(I,k)$ where $(I, k) \in \Sigma^{*} \times \mathbb{N}$ where $\Sigma$ is a finite set of alphabets and $\mathbb{N}$ is the set of natural numbers.

\begin{definition}
\label{defn:FPT}
A parameterized problem $\Pi \subseteq \Sigma^* \times \mathbb{N}$ is said to be {\em fixed-parameter tractable} (or {\em FPT} in short) if there exists an algorithm $\mathcal{A}$ which given an instance $(I, k)$ of $\Pi$, runs in $\oh(f(k) \cdot |I|^{c})$ time where $f: \mathbb{N} \rightarrow \mathbb{N}$ is a computable function depending only on $k$, and $c$ is a positive constant independent of $n$ and $k$.
\end{definition}

We denote the running time $\oh(f(k) \cdot |I|^{c})$ by the shorthand notation $\oh^{*}(f(k))$ where we supress the polynomial factors.
 We adopt the notations and symbols related to parameterized algorithms { from} the book \cite{cygan2015parameterized,downey2012parameterized}.
 
{
 \begin{definition}
 \label{defn:kernel}
 A parameterized problem $\Pi$ admits a {\em kernelization} (or a {\em kernel}) if there exists a polynomial-time algorithm $\mathcal{A}$ that given  an input instance $(I,k)$ of $\Pi$, constructs an equivalent instance $(I^{'},k^{'})$ of $\Pi$ such that $|I^{'}| + k^{'} \leq g(k)$ for some computable function $g(\cdot)$.
 \end{definition}}
 
It has been shown by Cai et al. \cite{cai1997advice} that a problem is in \textsf{FPT} if and only if there exists a kernelization.
If a parameterized problem $\Pi$ admits a kernelization algorithm, we also { say} that $\Pi$ admits a {\em kernel} (in short).
This function $g(\cdot)$ in Definition \ref{defn:kernel} denotes the {\em size} of the kernel.
If $g(k)$ is bounded by a function polynomial in $k$, then $\Pi$ is said to admit a {\em polynomial kernel}.
We describe the kernelization process by writing a number of reduction rules.
A reduction rule takes one instance (say $\mathcal{I}$) of $\Pi$ and generates the reduced instance (say $\mathcal{I}^{'}$) of $\Pi$.
We say a reduction rule is {\em safe} if the following condition holds: ``$\mathcal{I}$ is a \yes-instance if and only if $\mathcal{I}^{'}$ is a \yes-instance."
The efficiency of a kernel (or kernelization algorithm) is determined by the size of the kernel. 
Many parameterized problems are fixed-parameter tractable but do not admit polynomial kernels unless ${\sf NP} \subseteq {\sf coNP/poly}$ (see \cite{BodlaenderJK13,BodlaenderJK14,BougeretJS22,DomLS14,EinarsonGJMW23,FluschnikHNN18,Hermelin20}).
So, from the perspective of polynomial-time preprocessing, we look for kernels of polynomial size.

\paragraph{\textbf{Graph Parameters.}} 
 In parameterized complexity, the natural parameterization is the solution size.
However, several structural graph parameters have also been taken into account \cite{guo2004structural}.
In our paper, the natural parameter is the vertex cover of the forcing graph.
As mentioned in \cite{HarrisN24}, a vertex cover of size at most $k$ can be computed in $\oh^{*}(1.25284^{k})$ time, where $k$ is the size of the vertex cover. 
 Another important graph parameter is the \emph{modulator to $\mathcal{G}$} where $\mathcal{G}$ is a graph class. 
A subset of the vertices $S \subseteq V(G)$ is said to be a {\em modulator} to graph class $\mathcal{G}$ if $G - S \in {\GG}$\footnote{The modulator to a graph class $\GG$ is often referred as ``deletion set to $\GG$'' in the literature.}.
 
 
\subsection{Some Preliminary Algorithmic Results} 
\label{sec:basic-results}

In this section, we establish some classical complexity dichotomy results and some related parameterized complexity results for this problem. 
The first part of this section gives a proof that the problem is polynomial-time solvable when the forcing graph is a $2K_2$-free graph.
Since the following lemma is crucial to give both the classical and parameterized complexity results, we first define an annotated problem and discuss that it can be solved in polynomial time.

\defprob{\textsc{Ext-SPFG}}{
 A simple undirected graph $G(V,E)$, two distinct vertices $s,t \in V(G)$, a forcing graph $G_{f}(E, E^{'})$ and a vertex cover $S$ of $G_f$.
}{Find a subset $E^{*} \subseteq E$ with minimum number of edges such that
$S \subseteq E^*$, i.e. $E^*$ extends $S$, and the induced subgraph in $G$ by the edge set $E^{*}$ contains an $s$-$t$ path in $G$.
}

\medskip

\begin{lemma}
\label{lemma:extend-part-poly}
{\sc Ext-SPFG} can be solved in polynomial time.
\end{lemma}

\begin{proof}
Let $S$ be a vertex cover of $G_f$.
Note that $S$ is a set of edges in $G$, and let $G^* = G(V, S)$ be the subgraph induced by the edges of $S$.
If $s$ and $t$ are in the same connected component of $G^*$, then we are done, as we do not need to add any additional edge to get an $s$-$t$ path in $G$.
Otherwise, $s$ and $t$ are in different connected components of $G^*$.
Let us consider each of the connected components of $G^*$ one by one.

\medskip

We look at the connected components $C_1$ and $C_2$ of $G^*$ such that $s \in C_1$ and $t \in C_2$.
Observe that any optimal solution to {\sc Ext-SPFG} must be a shortest path between $u$ and $v$ such that $u \in C_1$ and $v \in C_2$.
Then, we perform {\em identification} of the vertex subset $V(C_1)$ by $x_{C_1}$ and  $V(C_2)$ by $x_{C_2}$ to get a new graph $\hat G$.
Then, we find a shortest path $\hat P$ from $\hat x_{C_1}$ to $\hat x_{C_2}$ in $\hat G$.
For every edge $uv \in \hat P$, if $u \notin V(G)$ or $v \notin V(G)$, then it must be that these endpoints are one of $u = x_{C_1}$ or $v = x_{C_2}$ or both.
If $u = x_{C_1}$ but $v \in V(G)$, then we replace the edge $uv$ by $u^*v$ such that $u \in C_1$ such that $u^*v \in E(G)$.
Similarly, if $v = x_{C_2}$ but $u \in V(G)$, then we replace $uv$  by $uv^*$ such that $v^* \in C_2$ and $uv^* \in E(G)$.
Finally, if $u = x_{C_1}$ and $v = x_{C_2}$, then we replace $u$ by $u^*$ and $v$ by $v^*$ such that $u^*v^* \in E(G)$.
We output $E(\hat P) \cup S$ as the optimal solution to the given instance of {\sc Ext-SPFG}.
Since this procedure iterates over all possible $u$-$v$ paths of the shortest possible number of edges from $u \in C_1$ to $v \in C_2$, the correctness of this algorithm follows.
\end{proof}

Using the above lemma, we provide a polynomial-time algorithm that computes a shortest path from $s$ to $t$ in $G$ when the forcing graph $G_f$ is $2K_2$-free.

{
\begin{lemma}
\label{Th:2_K_2}
The {\spfgp} can be solved in polynomial time if the number of minimal vertex covers of the forcing graph is upper bounded by a polynomial in the size of the primary graph.
\end{lemma}

\begin{proof}
Let $G$ be a primary graph with $n$ vertices, and $G_f$ be the forcing graph of $G$.
We assume that $|E(G)| = m$, hence $|V(G_f)| = m$.
Let $\mathcal{X}$ denote the set of all enumerated minimal vertex covers.
Then, due to a result by Tsukiyama et al. \cite{TsukiyamaIAS77}, the set of all minimal vertex covers can be enumerated in $\OO(n m |\mathcal{X}|)$-time.
As our primary graph $G$ is a simple graph, $m$ is upper bounded by $\OO(n^2)$.
As $|\mathcal{X}|$ is $m^{\OO(1)}$, the number of minimal vertex covers of $G_f$ is upper bounded by $n^{\OO(1)}$.
Now, we process each of the minimal vertex covers one by one.
Let $S$ be a minimal vertex cover of size $\ell$ in $G_f$, $V(S)$ denote the endpoints of all the edges of $S$, and $G^* = G[S]$ be the subgraph induced by the edges of $S$.
Now, our problem reduces to {\sc Ext-SPFG}, where apart from the primary graph $G$, we are also given a set $S \subseteq E(G)$.
For every such minimal vertex cover $S$ of $G_f$, we invoke the polynomial-time algorithm described in Lemma \ref{lemma:extend-part-poly} to compute a set $E^*$ with the minimum number of edges.
Then, output the specific edge set with the minimum number of edges.
As the number of such minimal vertex covers of $G_f$ is polynomially many and {\sc Ext-SPFG} can be solved in polynomial time, it follows that {\spfgp} can be solved in polynomial time.
\end{proof}
}


%

The above lemma illustrates that if the forcing graph is $2K_2$-free, then the optimization version of the {\spfgp} can be solved in polynomial time.
Our next result completes this picture by the following dichotomy.


\begin{theorem}
\label{corollary:class-SPFG-dichotomy}
{\spfgp} is polynomial-time solvable when the forcing graph is a $2K_2$-free graph.
\end{theorem}

\begin{proof}
Let $G$ be a primary graph with $n$ vertices, and $G_f$ be the forcing graph of $G$.
A connected graph is said to be $2K_{2}$-free if it does not contain a pair of independent edges as an induced subgraph.
It has been mentioned in a study by Farber \cite{farber1989diameters} that the complement of $2K_{2}$-free graphs (i.e., $C_{4}$-free graphs) has polynomially many maximal cliques.
Hence, if a graph with $n$ vertices is a $2K_{2}$-free graph, then it will have $n^{\OO(1)}$ many maximal independent sets.
Furthermore, the complement of a maximal independent set is a minimal vertex cover.
Now, the collection of all maximal independent sets of a $2K_{2}$-free graph can be enumerated in polynomial time (see Algorithm 3 in \cite{dhanalakshmi20162k2}).
So, the number of minimal vertex cover in any $2K_{2}$-free $n$-vertex graph is $n^{\OO(1)}$ and computable in polynomial time.
Note that the class of all $2K_2$-free graphs satisfies the premise of Lemma \ref{Th:2_K_2}.
Therefore, we invoke the algorithm by Lemma \ref{Th:2_K_2}, and this completes our argument that SPFG is polynomial-time solvable when the forcing graph is $2K_2$-free.
\end{proof}

After discussing the classical complexity of {\spfgp}, we move on to discuss the parameterized complexity of the same.
Since solution size is the most natural parameter, i.e., the number of edges in an optimal solution, we first prove that {\spfgp} is \fpt when parameterized by the solution size.
%
 For this purpose, we use an existing result by Damaschke et al. \cite{damaschke2006parameterized} stated in Proposition \ref{Th:1}.   

\begin{proposition}
\label{Th:1} \cite{damaschke2006parameterized}
Given a graph $G$ and positive integer $k$, all the vertex cover of $G$ of size at most $k$ can be enumerated in $\OO(m+2^{k}k^{2})$ time and in $\OO(2^{k})$ space.
\end{proposition}

\begin{theorem} \label{Th:2}
The {\spfgp} is fixed-parameter tractable and can be solved in $\oh((m+2^{k}k^{2})(m+n))$ time and $\oh^{*}(2^{k})$ space.
\end{theorem}

\begin{proof}
The steps involved in the algorithm are as follows.
 First, we invoke Proposition \ref{Th:1} to enumerate all the minimal vertex covers of $G_{f}$ of size at most $k$.
 Let $\mathcal{X}$ denote the set of all the enumerated minimal vertex covers of $G_f$.
 Each vertex cover in $\mathcal{X}$ represents a subset of edges of $G$.
 Given a subset $X \subseteq E(G)$, it follows from Lemma \ref{lemma:extend-part-poly} that we can find a set $E^*$ with a minimum number of edges that extends $X$ such that $G[X]$ has a path from $s$ to $t$.
 If $|E^*| \leq k$ for some minimal vertex cover $X$ of $G_f$, we output that the input instance is a \yes-instance.
 Otherwise, for every minimal vertex cover $X$ of $G_f$, if the Lemma \ref{lemma:extend-part-poly} outputs $E^*$ such that $|E^*| > k$, then output that the input instance is a \no-instance.
 Since there are at most $2^{k}n^{\OO(1)}$ minimal vertex covers of $G_f$ and for each such minimal vertex cover, {\sc Ext-SPFG} can be solved in polynomial-time, this problem can be solved in $2^k n^{\OO(1)}$-time.
\end{proof}

\section{Polynomial Kernels for {SPFG}}
\label{sec:kernel-soln-size}

In the previous section, we have discussed that {\spfgp} is fixed-parameter tractable when there are no restrictions on the { primary graph} $G$ and the forcing graph $G_f$, i.e,  $G$ and $G_f$ are arbitrary graphs.
This section is devoted to the kernelization complexity of the {\spfgp} problem when solution size is considered as the parameter.
%

As the edges of $G$ are the vertices of $G_f$, we define the following edge subsets of $G$.
\begin{itemize}
	\item We put an edge $e \in E(G)$ in $H$ if $deg_{G_f}(e) \geq k + 1$.
	\item We put $e \in E(G)$ in $L$ if $N_{G_f}(e) \subseteq H$.
	\item $R = E(G) \setminus (H \cup L)$.
\end{itemize}

Notice for any edge $e\in E(G)$, if $N_{G_f}(e) \subseteq H$, then $e \in L$.
{ Moreover, if $e \in H$, then $e$ must be part of any vertex cover of $G_f$ of size at most $k$ (if exists).}
Hence, we have the following observation.

\begin{observation}
\label{obs:isolated-in-I}
If $e$ is an isolated vertex in $G_f$, then $e \in L$.
{ Every vertex cover of $G_f$ of size at most $k$ must contain $H$.} 
\end{observation}

Now, we prove the following lemma, which will be one of the important parts in obtaining the kernel.

\begin{lemma}
\label{lemma:bounds-for-H-R}
If $I(G, G_f, s, t, k)$ is a \yes-instance then $|H| \leq k$ and $G_f[R]$ has at most $k^2$ edges.
\end{lemma}

\begin{proof}
By Observation \ref{obs:isolated-in-I}, an isolated vertex $e$ of $G_f$ is in $L$.
So, every $e \in R$ has at least one neighbor (with respect to $G_f$) in $R$.
As any $e \in V(G_f)$ with degree at least $k+1$ in $G_f$ is put in the set $H$, any $e \in R$ must have at most $k$ neighbors in $R$.
By our hypothesis, $I(G, G_f, s, t, k)$ is a \yes-instance.
Hence, there is $E^* \subseteq E(G)$ such that $|E^*| \leq k$ for every $(a, b) \in E(G_f)$, at least one of $a$ and $b$ must be in $E^*$.
If there is $e \in H \setminus E^*$, then at least $k+1$ edges have to be in $E^*$, that is a contradiction to the fact that $I(G, G_f, s, t, k)$ is a \yes-instance.
Hence, $H \subseteq E^*$ implying that $|H| \leq k$.
Consider the set $R$.
As every $e \in R$, at least one neighbor belongs to $R$ in $G_f$ and at most $k$ neighbors belong to $R$ in $G_f$.
Hence, the number of edges in $G_f$ that are incident to $R$ is at most $k^2$.
Therefore, the cardinality of  $V(R)$ is at most $2k^2$.
\end{proof}

Observe that the edges in $H$ are necessary for any solution of size at most $k$.
On the other hand, the role of the edges in $L$ is only to provide connectivity between $s$ and $t$ in $G$.
Let $E_k$ be the set of edges in $G_f[H \cup R]$. Recall,  $V(E_k)=\{u,v |~ e(u,v)\in E_k\}$.
For our convenience, we also add $s$ and $t$ into $V(E_k)$.
More formally, $V(E_k)=V(E_k) \cup \{s,t\}$ and let $Y = V(G) \setminus V(E_k)$.
We mark some vertices from $Y$ using the marking scheme {\sf MarkSPFG}($G, G_f, s, t, k$) that we describe as follows.

\paragraph{\underline{{\sf MarkSPFG}($G, G_f, s, t, k$)}}
\begin{itemize}
	\item For each unordered pair $(x,y)$ of vertices in $V(E_k)$ compute a shortest $x$-$y$ path, $P_{xy}$ via the internal vertices of $Y$ in $G$. 
	\item If $P_{x,y}$ has at most $k$ edges, then mark the edges of $P_{x,y}$.
	\item Else $P_{x,y}$ has more than $k$ edges. Then, do not mark any edge.
	\item Additionally, for every pair $x, y \in V(E_k)$, mark the edges of a shortest path $Q_{x, y}$ in $G$ when $|Q_{x,y}| \leq k$.
	\item { Finally, for every $e \in H$, we mark $\min\{|N_{G_f}(e) \cap L|, k+1\}$ edges from $N_{G_f}(e) \cap L$.
	We use $L_e$ to denote the set of all edges in $L$ that are marked for the edge $e \in H$.}
\end{itemize}

Let $E_t = \big(\bigcup\limits_{e \in H} L_e \big) ~\cup ~\big(\bigcup\limits_{x, y \in V(E_k)} (P_{x,y} \cup Q_{x,y})\big)$ be the set of marked edges of $G$ after the completion of the above marking scheme.
Consider $E_M = E_t \cup H \cup R$.
Our output graph is $G(V(E_M), E_M)$ where the vertex subset is $V(E_M)$, and the set of edges is $E_M$.
We use $G'$ in short to represent the output graph.
Consider the instance as $I(G', G_f[E_M],s,t,k)$. 
{
Our following observation ensures us that every solution of $I(G', G_f[E_M], s, t, k)$ of size at most $k$ must contain $H$.

\begin{observation}
\label{obs:H-criteria-reduced-graph}
\sloppy Let $I(G', G_f[E_M], s, t, k)$ be the instance obtained after invoking {\sf MarkSPFG}($G, G_f, s, t, k$).
Then, every solution of size at most $k$ of $I(G', G_f[E_M],s,t,k)$ must contain $H$. 
\end{observation}

\begin{proof}
We consider the marking procedure {\sf MarkSPFG}($G, G_f, s, t, k$).
In the last step of this procedure, we ensure that for every $e \in H$, if there are at least $k+1$ neighbors of $e$ in $L$, then $k+1$ neighbors of $L$ are marked.
If there are at most $k$ neighbors of $e$ in $L$, then all neighbors of $e$ that are in $L$ are marked.
In the latter case, observe that by definition of $H$, $e$ has at least $k+1$ neighbors in $G_f$ out of which at most $k$ neighbors are in $L$.
Other neighbors are in $H \cup R$.
By the choice of our output instance $I(G', G_f[E_M],s,t,k)$, $H \cup R \subseteq E_M$.
Hence, $k+1$ neighbors of $e$ are in $G_f$.
It implies that $e$ is part of any vertex cover of $G_f$ with at most $k$ vertices.
Therefore, every solution of size at most $k$ of $I(G', G_f[E_M],s,t,k)$ must contain $H$.
\end{proof}
}

Next, we prove the following lemma that is crucial for our polynomial kernelization.

\begin{lemma}
\label{Lemma:kernel-vc-fg-general}
The instance $I(G, G_{f}, s, t, k)$ is a \yes-instance if and only if $I(G', G_f[E_M], s, t, k)$ is a \yes-instance.
\end{lemma}

\begin{proof}
Let us first give the backward direction ($\Leftarrow$) of the proof.
First, assume that the instance $I(G', G_f[E_M], s, t, k)$ is a \yes-instance.
One can make a note that edges present in $G'$ are also in $G$ and in $G_f[E_M]$ as vertices.
Suppose that $G'$ contains a set of at most $k$ edges $E^*$ such that $G[E^*]$ has an $s$-$t$ path in $G'$ and $E^*$ is a vertex cover in $G_f[E_M]$.
{
By the procedure of {\sf MarkSPFG}($G, G_f, s, t, k$), for any edge $e \notin E_M$, $e$ is present in $L$.
Since $|E^*| \leq k$, it follows due to Observation \ref{obs:H-criteria-reduced-graph} that $H \subseteq E^*$}.
For every $e \in L$, $N_{G_f}(e) \subseteq H$.
Hence $E^*$ also { forms} a vertex cover of $G_f$.
Moreover, an $s$-$t$ path passing through a (proper) subset of edges in $E^*$ is also an $s$-$t$ path in $G$.
Hence, $I(G,G_f,s,t,k)$ is a \yes-instance.

Next, we focus on proving the forward direction ($\Rightarrow$). 
Assume that $I(G,G_{f},s,t,k)$ is a \yes-instance. 
Let $E^*$ be the solution to the instance $I(G,G_{f},s,t,k)$ and let $P$ be an $s$-$t$ path contained inside the graph induced by $G(V, E^*)$.
As $|E^*| \leq k$ is a solution to $I(G, G_f, s, t, k)$, { it follows from Observation \ref{obs:isolated-in-I} that} $H \subseteq E^*$.
If $E^* \subseteq E_M$, then $E^*$ is a solution to $I(G', G_f[E_M], s, t, k)$ and we are done.
In case some edge $e \in E^* \setminus E_M$ does not belong to any $s$-$t$ path in $G'$, then clearly such an edge { $e$ is unmarked by {\sf MarkSPFG}($G, G_f, s, t, k$) contained in $L$}.
We consider those subpaths (one at a time) $P^* \subseteq P$ that contain an edge $e \in E^* \setminus E_M$.
Observe that $P^*$ has at most $k$ edges and is an $x$-$y$ path in $G$ for some $x, y \in V(E_k) \cup \{s, t\}$.
But, we have marked a shortest path $\hat P^*$ from $x$ to $y$ in $G$ (via the vertices of $Y$ or in $G$ itself).
We just replace the edges of $P^*$ by $\hat P^*$.
As $|\hat P^*| \leq |P^*|$, this constructs an $s$-$t$ walk.
Similarly, for other subpaths also, we use the same replacement procedure and eventually construct an $s$-$t$ walk with at most $k$ edges in $G'$.
Let $\hat{E^*}$ denote the set of edges of $G'$ that are constructed by this replacement procedure.

Hence, we ensure that in the subgraph $G'[\hat{E^*}]$, there exists an $s$-$t$ path.
As $H \subseteq \hat{E^*}$ by our construction, $\hat{E^*}$ also forms a vertex cover of $G_f[E_M]$.
Finally, observe that the number of edges added during every replacement procedure is no more than the number of edges deleted during this replacement procedure.
Hence, $|\hat{E^*}| \leq |E^*|$.
Therefore, $\hat{E^*}$ is a solution to $I(G', G_f, s, t, k)$.
This completes the proof.
\end{proof}

Observe that for every pair of vertices in $V(E_k)$, we have marked a shortest path of length at most $k$ in $G$.
We are ready to prove our final theorem statement. 

\begin{theorem}
\label{Th:kernel-SP-FG}
The {\spfgp} admits a kernel with $\OO(k^5)$ vertices and edges.
\end{theorem}

\begin{proof}
We argue the correctness of the statement by counting the number of edges (vertices) at each step of the aforementioned preprocessing phase. 
First, we compute a partition of $V(G_f) = H \uplus R \uplus L$ as described.
From Lemma \ref{lemma:bounds-for-H-R}, we have that $H \cup R$ has at most $\oh(k^2)$ edges in $G_f$.
After that, we invoke the marking scheme {\sf MarkSPFG}($G, G_f, s, t, k)$ as described.
Observe that the marking scheme marks a shortest path of length at most $k$ for every pair of vertices $x, y \in V(E_k)$ and puts them in $E_M$.
Hence, $|E_M|$ is $\OO(k^5)$.
From Lemma \ref{Lemma:kernel-vc-fg-general}, { it follows that} $I(G, G_f, s, t, k)$ is a {\yes}-instance if and only if $I(G', G_f[E_M], s, t, k)$ is a {\yes}-instance.
We output $(G', G_f[E_M], s, t, k)$ as the output instance.
As $|E_M|$ is $\OO(k^5)$, the number of vertices in $G'$ is also $\OO(k^5)$.
Therefore, {\spfgp} admits a kernel with $\OO(k^5)$ vertices and edges.
\end{proof}

\section{Improved Kernels for Special Graph Classes}
\label{sec:improved-kernel-special-graph-class}

Consider an input instance $I(G, G_f, s, t, k)$ to {\spfgp}.
This section is devoted to kernelization algorithms when either $G$ or $G_f$ belongs to some special graph class.

\subsection{Improved Kernel when the primary graph is Planar}
\label{sec:forcing-graph-planar}

In this section, we discuss how we can obtain an improved kernelization upper bound when $G$ is a planar graph.
Note that the vertices of $G_f$ are the edges of $G$.
The primary graph $G$ has a special property, but the forcing graph $G_f$ can be an arbitrary graph.
What it means is that if $G$ is a simple graph with $n$ vertices, then $G$ can have at most $3n - 6$ edges.
Moreover, $G$ satisfies the Euler's formula.
But $G_f$ can be arbitrary.
We partition the vertices of $G_f$, i.e., the edges of $E(G)$ into $H, L$, and $R$ as before.
\begin{itemize}
	\item Put an edge $e \in E(G)$ into $H$ if $deg_{G_f}(e) > k$.
	\item Put $e \in E(G)$ into $L$ if $N_{G_f}(e) \subseteq H$.
	\item Define $R = E(G) \setminus (H \cup L)$.
\end{itemize}

It follows from Lemma \ref{lemma:bounds-for-H-R}, that if $I(G, G_f, s, t, k)$ is a {\yes}-instance, then $|H| \leq k$ and $G_f[R]$ has at most $k^2$ edges.
Since every vertex of $G_f[H \cup R]$ is incident to some edge, $H \cup R$ has $\OO(k^2)$ vertices.
Every vertex of $H \cup R$ is an edge in $G$.
It is clear that $H \cup R$ comprises to $\OO(k^2)$ edges in $G$, hence spans $\OO(k^2)$ vertices in $G$.
But, we do not yet have a concrete upper bound on the number of other edges of $G$ that are disjoint from the vertices of $H \cup R$. 
{ Let $E_k$ denote the set of edges in $G_f[H \cup R]$.
Hence, $V(E_k)$ is the set of vertices spanned by the edges of $G$ that are present in the set $H \cup R$ in addition to $s$ and $t$.
We assume $Y = V(G) \setminus V(E_k)$.
For every pair $\{x, y\}$ of $V(E_k)$, we define a boolean variable $J_{(\{x, y\})}$, { which} is true if there is a path from $x$ to $y$ in $G$ with internal vertices in $Y$ and $J_{(\{x, y\})}$ is false otherwise.} 
We exploit the following structural property of the planar graph.
A similar variant of the following was also proved independently by Wang et al. \cite{WangYGC13} and Luo et al. \cite{LuoWFGC13}.

\begin{lemma}
\label{lemma:G-number-of-other-edges}
There are at most $3|V(E_k)| - 6$ distinct pairs of vertices $\{x, y\}$ in $V(E_k)$ for which the boolean variable $J_{(\{x, y\})}$ is true.
\end{lemma}

\begin{proof}
The proof is based on the properties of planar graphs.
As $G$ is planar and $G[V(E_k)]$ is a subgraph, therefore, $G[V(E_k)]$ is also a planar graph.
Since $G[V(E_k)]$ is a simple graph (because $G$ is a simple graph), it follows from Proposition \ref{prop:planar-graph-edges} that  $G[V(E_k)]$ has at most $3|V(E_k)| - 6$ edges.
Therefore, if $G[V(E_k)]$ is allowed to have parallel edges, then there are at most $3|V(E_k)| - 6$ pairs of vertices for which there can be one edge or parallel edges.

Suppose for the sake of contradiction that there are more than $3|V(E_k)| - 6$ pairs of vertices $\{x, y\}$ for which the $J_{(\{x, y\})}$ value is true.
Then there are more than $3|V(E_k)| - 6$ pairs of vertices $\{x, y\}$ for which there is a path from $x$ to $y$ in $G$ using the vertices in $Y$.
Planar graphs are closed under contraction of edges.
If we contract all the edges of each of these paths until all the vertices of that path from $Y$ are removed, then this creates an additional edge between $x$ and $y$ in $G[V(E_k)]$.
But if we repeat this procedure for all such pairs of vertices, then this will result in having edges between more than $3|V(E_k)| - 6$  pairs of vertices in $G[V(E_k)]$.
Then, if we count all the parallel edges, there are more than $3|V(E_k)| - 6$ edges in $G[V(E_k)]$.
This will contradict the implication from Proposition \ref{prop:planar-graph-edges} that there can be at most $3|V(E_k)| - 6$ edges of $G$, having $G$ been a simple graph.
This completes the proof of the lemma.
\end{proof}


{
Our next step is to invoke the marking scheme {\sf MarkSPFG}($G, G_f, s, t, k$) as we did in Section \ref{sec:kernel-soln-size}.}
Let $E_t \subseteq E(G)$ be the set of all the edges that are marked by the procedure {\sf MarkSPFG($G, G_f, s, t, k$)} and $E_M = E_t \cup H \cup R$.
We output the graph $G(V(E_M), E_M)$ that is a subgraph (not necessarily induced) of $G$.
In short we represent this output graph by $G'$.
Then the following lemma holds, the proof of which is similar to the proof of Lemma \ref{Lemma:kernel-vc-fg-general}.

\begin{lemma}
\label{lemma:equivalence-spfg-g-planar}
The instance $I(G,G_f,s,t,k)$ is a \yes-instance if and only if $I(G', G_f[E_M],s,t,k)$ is a \yes-instance.
\end{lemma}

\begin{proof}
We first give the forward direction ($\Rightarrow$) of the proof.
Let $I(G, G_f, s, t, k)$ be a \yes-instance and $E^*$ be a set of at most $k$ edges from $G$ such that there is a path from $s$ to $t$ in $G(V, E^*)$ and $E^* \subseteq V(G_f)$ forms a vertex cover of $G$.
If $E^* \subseteq E_M$, i.e. the set of edges from the subgraph $G[E_M]$, then we are done.
We assume that there is some edge $e \in E^*$ such that $e \notin G[E_M]$.
Note that any such edge must be an edge $e \in L$.
Since $|E^*| \leq k$, it follows that $H \subseteq E^*$.
If $e$ is not in an $s$-$t$ path in $G'$, then $e$ must be contained in $L$.
As $H \subseteq E^*$, removing $e$ from the solution $E^*$ does not violate the property that $E^* \setminus \{e\}$ is a vertex cover in $G_f[E_M]$.
The interesting case is when $e$ is in an $s$-$t$ path $P^*$ in $G$.
Since there can be more than one such edge, we consider each of those subpaths $P^* \subseteq P$ (one at a time).
Let $P^*$ be a path from $x$ to $y$ for some $x, y \in V(E_k)$.
Since $|P^*| \leq k$, it follows that we have marked the edges of a shortest path $P_{x, y}$ with at most $k$ edges from $x$ to $y$ in $G$.
Either $P_{x, y}$ uses only the vertices of $Y$, or it is a shortest path from $x$ to $y$ in $G$.
Hence, it follows that $|E(P_{x, y})| \leq |E(P^*)|$.
Since all the edges of $P_{x, y}$ are marked, we replace the subpath $P^*$ by $P_{x, y}$.
This replacement procedure does not affect the connectivity between $s$ and $t$ and does not increase the number of edges.
We repeat this process for every such edge that is in $P^*$ but not in $G'$.

{ For the backward ($\Leftarrow$) of the proof, let the $I(G', G_f[E_M],s,t,k)$ be a {\yes}-instance.
What is crucial here is that the edges in $G'$ are also present in $G$.
Hence, the edges of $G'$ are present in $G_f[E_M]$ as vertices.
If $G'$ contains a set of at most $k$ edges $E^*$ such that $G'[E^*]$ has an $s$-$t$ path $E^*$ is a vertex cover of $G_f[E_M]$.
Since $|E^*| \leq k$, due to Observation \ref{obs:H-criteria-reduced-graph}, $H \subseteq E^*$.
It means that for every $e \in L$, $N_{G_f}(e) \subseteq H$.
By the procedure of {\sf MarkSPFG}($G, G_f, s, t, k$), any edge $e \notin E_M$ is present in $L$.
Hence, $E^*$ also forms a vertex cover of $G_f$.
In particular, a path between $s$ and $t$ passing through a (proper) subset of edges in $E^*$ is an $s$-$t$ path in $G$.
Hence, $I(G, G_f, s, t, k)$ is a {\yes}-instance.
}
This completes the proof.
\end{proof}


Using the above lemma, we have the following result.

\begin{theorem}
\label{thm:SPFG-planar-input}
{\spfgp} admits a kernel with $\OO(k^3)$ vertices when $G$ is a planar graph.
\end{theorem}

\begin{proof}
Let $I(G, G_f, s, t, k)$ be an input instance to the {\spfgp} such that $G$ is a planar graph.
In the first stage of our kernelization algorithm, we partition the vertices of $G_f$, i.e., $V(G_f) = H \uplus R \uplus L$ as above.
If $I(G, G_f, s, t, k)$ be a \yes-instance, then $|H| \leq k$ and $G_f[H \cup R]$ has $\OO(k^2)$ edges.
Since there is no vertex in $H \cup R$ that is isolated in $G_f$ and $G_f[H \cup R]$ has $\OO(k^2)$ edges, it follows that $G_f[H \cup R]$ has $\OO(k^2)$ vertices.

During the second phase of our kernelization algorithm, we call upon the marking algorithm denoted as ${\sf MarkSPFG}(G, G_f, s, t, k)$ and output the graph $G'(V(E_M), E_M)$.
Note that $H \cup R$ has $\OO(k^2)$ vertices.
We only focus on the unordered pairs $(x, y)$ of vertices from $V(E_k)$ such that there is a path from $x$ to $y$ using the edges of $G$.
It follows from Lemma \ref{lemma:G-number-of-other-edges} that there are at most $3|V(E_k)| - 6$ such pairs for which there is a path between them in $G$.
For each pair of vertices from $V(E_k)$, we mark the edges of a shortest path of length at most $k$.
Independently, for every such pair, we mark the edges of another shortest path with at most $k$ edges.
This ensures that the number of edges that are marked by {\sf MarkSPFG}($G, G_f, s, t, k$) is $\OO(k^3)$.
We keep only the edges of $G$ that have both endpoints in $V(E_k)$ or the edges that are marked using this marking scheme.
As $|V(E_k)|$ is $\OO(k^2)$, the number of vertices that are spanned by the above-mentioned edges is also $\OO(k^3)$.
We delete the other vertices of the graph that are not spanned by the above set of edges.
This completes the proof that if the primary graph is planar, then {\spfgp} admits a kernel with $\OO(k^3)$ vertices.
\end{proof}

\subsection{Improved Kernels when the Forcing Graph is in Special Graph Classes}
\label{sec:forcing-graph-restricted}

In this section, we explain how the kernelization upper bound proved in Theorem \ref{Th:kernel-SP-FG} can be improved when the forcing graph belongs to some special graph classes.
Note from the proof of Theorem \ref{Th:kernel-SP-FG} that the preprocessing algorithm needs to reduce the vertices of two different categories separately.
The first category is the {\em hitting part in $G_f$} and the second category is the {\em $(s, t)$-connectivity in $G$}.
In the previous section, we first proved Lemma \ref{lemma:bounds-for-H-R}, which states that for every \yes-instance, the number of vertices that are required to intersect the edges of $G_f$ is $\OO(k^2)$.
But in such a case, $G_f$ is any arbitrary graph.
The subsequent lemma statements illustrate that when $G_f$ is in some special graph classes, e.g., cluster graphs, { or} bounded degree graphs, we can ensure that $\OO(k)$ vertices are sufficient for the hitting part, i.e., to hit all the edges of $G_f$.
But before proving them, let us provide some intuitions.
When $G_f$ is a cluster graph, every connected component of $G_f$ is a clique.
If $C$ is a clique in a graph $G_f$, then at least $|C| - 1$ vertices of $C$ are part of any vertex cover of $G_f$.
In case the input instance is a \yes-instance, the number of vertices in a feasible solution is at most $k$.
Hence, the following lemma statements hold.

\begin{lemma}
\label{lemma:G_f-is-cluster}
If $G_f$ is a cluster graph and $I(G, G_f, s, t, k)$ is a \yes-instance, then $G_f$ has at most $2k$  vertices that are not isolated in $G_f$.
If $G_f$ is a bounded degree graph with maximum degree at most $\eta$ and $I(G, G_f, s, t, k)$ is a \yes-instance, then $G_f$ has at most $2k\eta$ vertices that are not isolated in $G_f$.
\end{lemma}

\begin{proof}
Let $I(G, G_f, s, t, k)$ be a \yes-instance and $G_f$ be a cluster graph.
Observe that for every connected component $C$ of $G_f$, all but one vertex from $C$ must be part of any (optimal) solution.
Clearly, if the number of connected components of $G_f$ that have an edge is more than $k$, then any feasible solution to the instance $I(G, G_f, s, t, k)$ must have more than $k$ vertices implying that $I(G, G_f, s, t, k)$ is a \no-instance.
Therefore, there are at most $k$ connected components in $G_f$ that have at least an edge.
Since any feasible solution $X^*$ to $I(G, G_f, s, t, k)$ also must have all but at most one vertex from each such connected component $C$, it must be that $|C \cap X| \geq |C| - 1$.
Therefore, $|C| \leq |C \cap X| + 1$.
Additionally, for every connected component $C$ of $G_f$ that has an edge, $|C \cap X| \geq 1$.
Hence, $|C| \leq |C \cap X| + 1 \leq |C \cap X| + |C \cap X|$.
Hence, if $C$ is a connected component of $G_f$ that has an edge, then $|C| \leq 2|C \cap X|$.
Subsequently, if $\ccC$ denotes the set of all components of $G_f$ that have an edge, then 

$$\sum\limits_{C \in \ccC} |C| \leq 2\sum\limits_{C \in \ccC} |C \cap X|$$

Note that $\sum\limits_{C \in \ccC} |C \cap X| \leq |X|$.
As $|X| \leq k$, therefore, $\sum\limits_{C \in \ccC} |C| \leq 2k$.
Hence, $G_f$ can have at most $2k$ vertices that are not isolated.

On the other hand, let $G_f$ be a bounded degree graph having maximum degree at most $\eta$.
If $I(G, G_f, s, t, k)$ is a \yes-instance, then having a vertex cover $X^*$ of size at most $k$ of $G_f$ ensures that $G_f$ can have at most $k\eta$ edges.
It means that there are at most $2k\eta$ vertices that are not isolated in $G_f$.
\end{proof}

Now onward, we proceed assuming that the given instances are {\yes}-instances for both graph classes. The above lemma ensures that the number of non-isolated vertices in $G_f$ is bounded by $2k$ when $G_f$ is a cluster graph and by $k\eta$ when $G_f$ is a bounded degree graph with maximum degree at most $\eta$.
Moreover, if $G_f$ has at most $k\eta$ vertices for some constant $\eta$, it holds that $G$ has at most $k\eta$ edges.
But it is possible that there are some vertices that are isolated in $G_f$, but they are edges in $G$.
A subset of such edges can be useful in providing connectivity between $s$ and $t$ in $G$.
Let $E_k \subseteq V(G_f)$ be the set of vertices that are not isolated in $G_f$, and that is the set of edges in $G$.
{Let $V(E_k)$ denote the set of vertices of $G$ that are the endpoints of these edges of $E_k$.
In addition, we also add $s$ and $t$, hence $V(E_k) = V(E_k) \cup \{s, t\}$}.
Let $Y = V(G) \setminus V(E_k)$.
We use a marking procedure {\sf MarkSPFG}($G, G_f, s, t, k$) as before, but describe it for the sake of completeness.


Let $E_t$ be the set of all the edges that are marked by the { marking scheme {\sf MarkSPFG}$(G, G_f, s, t, k)$ as we did in Section \ref{sec:kernel-soln-size}}, and let $E_M = E_t \cup E_k$.
{Our output graph is $G(V(E_M), E_M)$.}
For the simplicity of presentation, we use $G'$ to denote our output graph.
We have the following lemma.
There are some similarities in arguments with the proof arguments of Lemma \ref{lemma:equivalence-spfg-g-planar}, but there are some differences in arguments as well.
So, we present the proof here for the sake of completeness.

\begin{lemma}
\label{lemma:vertex-bound-special-forcing-graph}
The input instance $I(G, G_f, s, t, k)$ is a \yes-instance if and only if the output instance $I(G', G_f[E_M], s, t, k)$ is a \yes-instance.
\end{lemma}

\begin{proof}
We first give the backward direction ($\Leftarrow$) of the proof.
Let $I(G', G_f[E_M], s, t, k)$ be a \yes-instance and $X^*$ be a set of at most $k$ edges that is a solution to this instance.
Note that in the subgraph $G'(V, X^*)$, there is already a path from $s$ to $t$.
If we can prove that $X^*$ is a vertex cover of $G_f$, then we are done.
What is crucial here is that no vertex from $V(E_k)$ was deleted in $G'$.
Therefore, $X^*$ must already be a vertex cover of $G_f$, implying that $I(G, G_f, s, t, k)$ is a \yes-instance.

Next, we give the forward direction $(\Rightarrow)$ of the proof.
Let $I(G, G_f, s, t, k)$ be a \yes-instance and $X^*$ be a feasible solution to this instance having at most $k$ edges from $G_f$.
What is crucial here is that $E(G_f) = E(G_f[E_M])$.
It means that no edge from $G_f$ was deleted in the graph $G_f[E_M]$.
\end{proof}

\begin{theorem}
\label{thm:SPFG-other-special}
{\spfgp} admits a kernel with $\OO(k^3)$ vertices when $G_f$ is either a cluster graph or a graph with bounded degree.
\end{theorem}

\begin{proof}
 As $G_f$ admits a vertex cover of size at most $k$, and either $G_f$ is a cluster graph or the maximum degree is bounded by a fixed value $\eta$, the number of edges in $G_f$ is $\OO(k^2)$.
By definition $E_k$ denote the set of vertices of $G_f$ that are not isolated and we added $s$ and $t$ into $V(E_k)$.
Hence, $|V(E_k)|$ is $\OO(k)$.
Recall, the marking scheme marks a shortest path of length at most $k$ for every pair of vertices $x, y \in V(E_k)$ and puts them in $E_M$. Hence, $|E_M|$ is $\OO(k^3)$.
Moreover, from Lemma \ref{lemma:vertex-bound-special-forcing-graph}, $I(G, G_f, s, t, k)$ is a \yes-instance if and only if $I(G[E_M], G_f[E_M], s, t, k)$ is a \yes-instance.
Therefore, {\spfgp} admits a kernel with $\OO(k^3)$ vertices when the forcing graph is either a cluster graph or a graph with a bounded degree.
\end{proof}

Observe that $2$-ladders or $3$-ladders are all subclasses of graphs of degree at most one and two, respectively.
Hence, our result also complements the APX-hardness result of Darmann et al. \cite{darmann2009determining} with the following consequence of the above theorem.

\begin{corollary}
{\rm SPFG} admits a kernel with $\oh (k^3)$ vertices when the forcing graph is a collection of disjoint paths and cycles.
\end{corollary}

\section{Structural Parameterization of SPFG}
In this section, we provide a result on the structural parameterization of the {\spfgp}.
We consider a few structural parameters of the forcing graph.
First, we consider when the deletion distance ($k$) to $2K_2$-free graph of $G_f$.
We restate the formal definition of {\sc SPFG-$2K_2$-Free-Deletion} problem.

\defparprob{\textsc{SPFG-$2K_2$-Free-Deletion}}{
 A simple, undirected graph $G(V,E)$, two distinct vertices $s,t \in V(G)$, a forcing graph $G_{f}(E, E^{'})$, a set $X \subseteq V(G_f)$ such that $G_f - X$ is $2K_2$-free and an integer $\ell$.
}{$|X|$}{Is there a subset $E^* \subseteq E$ of at most $\ell$ edges such that the subgraph $G(V, E^{*})$ contains an $s$-$t$ path in $G$, and $E^{*}$ forms a vertex cover in $G_{f}$?
}

Now, we prove that {\sc SPFG-{$2K_2$}-Free-Deletion} is fixed-parameter tractable.
We need the following lemma to prove the result.

\begin{lemma}
\label{lemma:spfg-2K2-del-enum}
Given an instance $(G, G_f, X, s, t, \ell)$ to the {\sc SPFG-{$2K_2$}-Free-Deletion} problem, the set of all minimal vertex covers of $G$ can be enumerated in $2^{|X|}n^{\OO(1)}$-time.
\end{lemma}

\begin{proof}
Let $(G, G_f, X, s, t, \ell)$ to the {\sc SPFG-{$2K_2$}-Free-Deletion} problem.
By hypothesis, the forcing graph $G_f - X$ is a $2K_2$-free graph.
The enumeration algorithm works as follows.
The first step is to guess a vertex subset $X_1 \subseteq X$ such that $X_1$ is a vertex cover of $G_f[X]$.
For every such vertex subset $X_1 \subseteq X$, note that the vertices of $X \setminus X_1$ are not picked into the solution.
Therefore, every neighbor of $X \setminus X_1$ that is in $G_f - X$ must be part of a (minimal) vertex cover that intersects $X$ in $X_1$.
{Let $Y_1 = N_{G_f}(X \setminus X_1) \setminus X$, the neighbors of $X \setminus X_1$ in $G_f - X$.}
As $2K_2$-free graphs are hereditary and $G_f - X$ is a $2K_2$-free graph, it must be that $G_f - (X \cup Y_1)$ is also a $2K_2$-free graph.
Due to Algorithm 3 from Dhanalakshmi et al. \cite{dhanalakshmi20162k2}, the set of all minimal vertex covers of $G_f - (X \cup Y_1)$ can be computed in polynomial time.
We repeat this procedure for every subset $X_1 \subseteq X$ such that $X_1$ is a vertex cover of $G_f[X]$.
Since there are $2^{|X|}$ possible choices of subsets of $X$ and the number of minimal vertex covers of $G_f - (X \cup Y_1)$ is polynomially many, the collection of all minimal vertex covers of $G_f$ can be enumerated in polynomial time.
\end{proof}

Now, we give a proof of our result using the above lemma.

\begin{theorem}
\label{thm:2K2-free-FPT-proof}
{\sc SPFG-{$2K_2$}-Free-Deletion} admits an algorithm that runs in $2^{|X|}n^{\oh(1)}$-time.
\end{theorem}

\begin{proof}
Similar as before, the first step is to enumerate the set of all minimal vertex covers of $G_f$ of size at most $\ell$.
We invoke Lemma \ref{lemma:spfg-2K2-del-enum} to enumerate all the minimal vertex covers in $2^{|X|}n^{\oh(1)}$-time.
Then, for every minimal vertex cover $X$ of $G_f$, we focus on the vertices $V(X)$ spanned by the set of edges in $X$ and check if there is a path from $s$ to $t$ in $G$ using at most $\ell - |X|$ edges.
Since this procedure takes polynomial time, {\sc SPFG-{$2K_2$}-Free-Deletion} admits an algorithm that runs in $2^{|X|}n^{\oh (1)}$-time.
\end{proof}

Observe that our algorithm heavily uses the characteristics that {\sc SPFG-${\GG}$-Deletion} is FPT when ${\GG}$ is a $2K_2$-free graph.
In particular, the proof idea is similar to the proof of Lemma \ref{Th:2_K_2} and exploits that the collection of all the minimal vertex covers in $2K_2$-free graph can be enumerated in polynomial time.
On the other hand, if SPFG is {\npc} when the forcing graph is in ${\GG}$ for some other graph class, then {\sc SPFG-{$\GG$}-Deletion} becomes {\npc} even when the size of the modulator to $\GG$ is a fixed-constant.

\section{Conclusion and Open Problems}
In this paper, we have initiated the study of {\spfgp} under the realm of parameterized complexity.
One natural open problem is to see if our kernelization results for {\spfgp} can be improved, i.e., can we get a kernel with $\oh(k^4)$ vertices for SPFG when both $G$ and $G_f$ are arbitrary graphs?
We strongly believe that those results can be improved, but some other non-trivial techniques might be necessary.
Independently, it would be useful to have a systematic study of this problem under positive disjunctive constraints containing three (or some constant number of) variables.
Additionally, from the perspective of kernelization complexity, we leave the following open problems for directions of future research.
\begin{itemize}
	\item Can we get a kernel with $\oh (k^4)$ vertices for SPFG when the primary graph $G$ is arbitrary graph but the forcing graph $G_f$ is a graph of degeneracy $\eta$? Our results only show that if the forcing graph is of bounded degree, then we can get a kernel with $\oh(k^3)$ vertices.
	In fact, even if $G_f$ is a forest, it is unclear if we can get a kernel with $\oh(k^3)$ or $\oh(k^4)$ vertices.
	\item What happens to the kernelization complexity when $G_f$ is an interval graph while $G$ is an arbitrary graph? Can we get a kernel with $\oh(k^3)$ vertices in such a case?
	\item Finally, can we generalize our result of Theorem \ref{thm:SPFG-planar-input} when $G$ is a graph of bounded treewidth or a graph of bounded degeneracy?
\end{itemize} 


\bibliographystyle{plain}

\end{document}